\documentclass[11pt]{article}
\usepackage[a4paper,top=1.25in,bottom=1.25in,left=1.25in,right=1.25in]{geometry}
\usepackage{amsmath,amsthm,setspace,pgfplots,tikz}
\usepackage[T1]{fontenc}
\usepackage[utf8]{inputenc}
\usepackage[affil-it]{authblk}
\usepackage[backend=biber,giveninits=true, style=authoryear]{biblatex}
\usepackage{float}
\usepackage{booktabs}
\usepackage{xurl}
\usetikzlibrary{decorations.pathreplacing}

\usepgfplotslibrary{fillbetween}
\pgfplotsset{compat=1.18}

\pgfmathdeclarefunction{Xsym}{2}{%
  \pgfmathparse{(2/(1-#2))*(1+#1)^2}%
}

\pgfmathdeclarefunction{hatf}{3}{%
  \pgfmathparse{(1+#3)*(1+#2*#1)-1}%
  }

\pgfmathdeclarefunction{barf}{4}{%
  \pgfmathparse{(Xsym(#1,#3)/(1+#2*#1))*(1-1/(2*(1+#4)))-1}%
}

\pgfmathdeclarefunction{psip}{3}{%
  \pgfmathparse{Xsym(#1,#3)/(2*(1+#2*#1)^2)-1}%
}

\pgfmathdeclarefunction{tssf}{4}{%
  \pgfmathparse{max(hatf(#1,#2,#4),barf(#1,#2,#3,#4))}%
}

\pgfmathdeclarefunction{psix}{3}{%
  \pgfmathparse{%
    (Xsym(#1,#3)
    + sqrt(max(0,(Xsym(#1,#3))^2
    - 2*Xsym(#1,#3)*(1+#2*#1)^2)))
    /(2*(1+#2*#1)^2)-1%
  }%
}

\pgfmathdeclarefunction{psirunhat}{3}{%
  \pgfmathparse{(1+#3)/(1+#2*#1)-1}%
}

\pgfmathdeclarefunction{psirunbar}{4}{%
  \pgfmathparse{1/(2*(1-(1+#4)/(Xsym(#1,#3)/(1+#2*#1))))-1}%
}

\pgfmathdeclarefunction{psirunf}{4}{%
  \pgfmathparse{min(psirunhat(#1,#2,#4),psirunbar(#1,#2,#3,#4))}%
}

\pgfplotsset{
  thresholdaxis/.style={
    width=0.92\textwidth,
    height=0.56\textwidth,
    grid=both,
    grid style={line width=.1pt, draw=gray!25},
    major grid style={line width=.2pt, draw=gray!35},
    tick label style={font=\small},
    label style={font=\small},
    title style={font=\small},
    legend style={font=\small, draw=none, fill=none},
    scaled ticks=false,
  },
}

\theoremstyle{plain}
\newtheorem{proposition}{Proposition}
\newtheorem{theorem}{Theorem}
\newtheorem{lemma}{Lemma}

\theoremstyle{definition}
\newtheorem{remark}{Remark}

\newenvironment{keywords}
{\begin{trivlist}\item[]{\bfseries Keywords:} }
{\end{trivlist}}

\newenvironment{jel}
{\begin{trivlist}\item[]{\bfseries JEL Codes:}\ }
{\end{trivlist}}

\begin{document}

\pagestyle{plain}
\setstretch {1.25}

\title{Surrender runs}

\author[a]{Andreas L\"offler}
\author[a]{Stefan Steins\thanks{Corresponding author. Email: stefan.steins@fu-berlin.de\\
This paper was written while the second author was a visiting researcher at the Deutsche Bundesbank. The authors thank Till Förstemann for helpful comments and discussion.}}

\affil[a]{\small Freie Universit\"at Berlin, 14195 Berlin, Germany}
\date{August 21, 2026}

\maketitle

\begin{abstract}
\noindent Rising interest rates can expose life insurers to surrender risk by reducing the market value of their assets and raising policyholders' outside returns. This paper develops a minimal model of an insurance-specific run mechanism in which strategic interaction arises because early surrenders can deplete the asset pool backing continuation values. We analyze a strategic surrender game in which a run is an equilibrium outcome determined jointly by asset values, contractual surrender claims, and payoff-dependent continuation benefits. The model yields closed-form interest-rate thresholds for fundamentals-driven and self-fulfilling runs. Thinner capitalization weakly lowers the joint-surrender cutoff. Continuation benefits raise surrender thresholds but can also create strategic complementarity.
\end{abstract}

\begin{jel}
G22, G01, E43, C72, G32
\end{jel}


\newpage



\section{Introduction}\label{sec:introduction}

Rising market interest rates increase life-insurance surrender incentives. Most life-insurance contracts contain an embedded surrender option. As market rates rise, the market value of insurers' fixed-income assets falls, while contractual surrender claims are fixed or adjust only slowly. At the same time, the outside return available to policyholders rises, making surrender more attractive and increasing the risk of mass surrenders. This has potential financial stability implications because life insurers hold around 20\% of global bonds and 30\% of credit investments \parencite[p.~21]{IMFOc.2021}.

\textcite{Feodoria.2015} identify the interest-rate level at which an insurer's marked-to-market assets cease to cover total contractual surrender claims. Beyond this rate, a marginal surrender with full payout of the contractual claim depletes the assets backing the remaining contracts, making insurers vulnerable to surrender waves. However, this underfunding threshold need not coincide with the rate at which surrender becomes optimal for policyholders. A surrender run is an equilibrium outcome determined jointly by asset values, contractual surrender claims, and surrender frictions. Surrender thresholds therefore generally exceed the rate at which aggregate surrender claims become underfunded.

The ECB's 2022--2023 monetary tightening cycle and associated bond market repricing make the distinction salient. In 2022, the Deutsche Bundesbank warned of the risk of a wave of policy lapses should yields on ten-year Bunds exceed \(3\%\), identifying this rate as the aggregate underfunding threshold for the German market \parencite[pp.~62--63]{bundesbank2022fsr}. On 14 September 2023, the ECB raised its key policy rates by 25 basis points \parencite{ecb2023sep_mpd}. The official reference yield on the ten-year Federal bond subsequently reached \(2.97\%\) in October, while yields on longer-maturity Federal securities briefly exceeded \(3\%\) \parencite[Table~I.3]{bundesbank2023cmi_nov}. Nevertheless, no broad-based lapse wave materialized in Germany; during the 2024--2025 reporting period, lapse rates fell below \(5\%\) \parencite[p.~19]{afs2025twelfth}.

European evidence points in the same direction, with Italy as the main exception among large insurance markets. EIOPA reports that the median lapse rate in the EEA was stable at around \(3.0\%\) in 2022 and \(3.1\%\) in 2023 \parencite{EIOPAFSR2023,EIOPAFSR2024}. By contrast, IVASS reports that 2023 surrenders in Italy increased by EUR~32 billion, or \(59\%\) relative to 2022, with a larger increase for contracts distributed mainly through banks \parencite{IVASSQuaderno31}. These contrasting outcomes suggest that balance-sheet vulnerability alone is insufficient to explain mass-surrender risk.

Subsequent German assessments explain why surrender thresholds may exceed the corresponding underfunding thresholds. In its 2023 Financial Stability Review, the Bundesbank revised its adjusted median critical rate to \(4.3\%\) (\cite[pp.~44--48]{bundesbank2023fsr}; see also \cite[p.~20]{afs2024eleventh}) and noted: ``Since a policy lapse entails transaction and information costs, and because supplementary insurance policies covering occupational disability and the like are lost, the critical interest rate is likely to be higher from the perspective of many policyholders'' \parencite[pp.~45--46]{bundesbank2023fsr}. The Bundesbank's 2025 report cites survey evidence that the risk would become substantial only around a \(6\%\) yield on the ten-year Bund \parencite[p.~85]{bundesbank2025fsr}.

This evidence motivates our contribution. We develop a parsimonious structural model in which surrender can reduce the assets backing continuation payoffs. Policyholder payoffs are determined by contractual claims, the market value of an insurer's assets, pro-rata recovery, and continuation benefits. Since benefits like insurance cover against life risks scale with the realized continuation payoff, they simultaneously raise continuation value and amplify the effect of asset dilution, thereby generating strategic complementarity. The model yields closed-form interest-rate thresholds that separate underfunding from equilibrium surrender incentives and classify run regions as fundamentals-driven or self-fulfilling.

The model primitives can be linked to observable market, contractual, and regulatory quantities. The equity parameter \(E\) is interpreted as a reduced-form own-funds buffer, the profit-participation parameter \(\lambda\) maps asset returns into credited contractual growth, and the interim market rate \(y_1\) corresponds to the one-period yield equivalent of a market-value asset shock. The continuation-benefit parameter \(\psi\) represents the payoff-dependent value of retained contract features.

We obtain two main existence results. Theorem~\ref{theorem1} characterizes when joint continuation and joint surrender coexist as Nash equilibria. Theorem~\ref{theorem2} gives a weaker existence condition under which joint surrender is the unique Nash equilibrium. The model yields transparent comparative statics. Proposition~\ref{prop:comp_stat_surrender_cutoff} shows that the strict joint-surrender cutoff is weakly increasing in the own-funds buffer and strictly increasing in the continuation-benefit parameter. To illustrate the economic magnitude of these effects, we present a parametrization. The numerical examples show that a weakly capitalized insurer is more exposed to mass-surrender risk, while strong continuation benefits disincentivize surrender even when asset values are significantly impaired. Compared to fundamentals-driven runs, self-fulfilling runs occur only at relatively high continuation-benefit levels and require large effective rate shocks.

The paper connects to three strands of literature. It adds directly to recent work on life insurer liquidity, critical rate thresholds, and aggregate lapse risk. \textcite{Feodoria.2015} define the rate at which a positive rate shock can make life insurers financially vulnerable if a wave of lapses occurs. \textcite{Forstemann.2021} develops a related frictionless surrender model around the critical-rate framework. \textcite{Chang.Schmeiser.2022} study surrender and liquidity risk jointly and show that liquidity constraints can increase surrender risk. \textcite{Kubitza.2025} document empirically that higher interest rates raise surrender rates and use a calibrated model to study the resulting liquidity and portfolio-rebalancing effects. \textcite{Koijen.2024} model heterogeneous exposure to common lapse factors related to credit spreads, unemployment, and interest rates. 

Second, the paper relates to the literature on surrender options and participating life-insurance contracts, which typically studies surrender, bonus, and premium features as valuation or optimal-stopping problems \parencite{Albizzati.Geman.1994,Grosen.2000,Bacinello.a.2003,Gatzert.Schmeiser.2008,Schmeiser.2011}. A broader literature on policyholder behavior incorporates taxes, transaction costs, incomplete or boundedly rational exercise, stochastic surrender intensity, non-financial motives, behavioral lapse-based pricing, surrender contagion, and heterogeneous policyholder behavior \parencite{DeGiovanni.2010,Eling.Kochanski.2013,Li.Szimayer.2014,Bauer.2017,Russo.2017,Gottlieb.Smetters.2021}. Within this strand, a closely related paper is by \textcite{Cheng.2023}, who model surrender contagion in a heterogeneous pool with financially guided and herd-driven policyholders. 

Third, the paper connects life-insurance contract design to the bank-run and global-games literature \parencites{Diamond.Dybvig.1983}{Carlsson.1993}{Morris.2002}{Goldstein.2005}. A conceptually related paper is by \textcite{Foley.2020}, who identify a self-fulfilling component in the run by institutional investors on puttable funding-agreement-backed securities issued by U.S. life insurers during the 2007--2008 crisis.

The remainder of the paper is organized as follows. Sections \ref{model}--\ref{runs} present the model, surrender frictions, and equilibrium run thresholds. Section \ref{compstat} derives comparative statics. Section \ref{parametrisation} maps the primitives into observable magnitudes and provides numerical illustrations. Section \ref{conclusion} concludes.

\section{Model}\label{model}

We extend the framework of \textcite{Forstemann.2021}. Consider a life insurer with two policyholders, indexed by $i\in\{A,B\}$. Time is discrete, with dates $t=0,1,2$. The insurer has initial equity $E$ and pays no dividends before date $t=2$. Let $y_0$ denote the annualized yield on a default-free zero-coupon bond from date $t=0$ to date $t=2$, and let $y_1$ denote the yield on the default-free zero-coupon bond from date $t=1$ to date $t=2$.\footnote{We assume throughout that $1+y_t>0$ for $t\in\{0,1\}$.}

At date $t=0$, policyholder $i$ opens a life-insurance policy with premium $S_{i,0}>0$. The contractual claim associated with policy $i$ at date $t$ is
\[
    S_{i,t}:=S_{i,0}(1+\lambda y_0)^t,
    \qquad t\in\{0,1,2\},
\]
where $\lambda\in(0,1)$ is the policyholder's profit-participation parameter.

The insurer invests initial equity and all premiums in the default-free zero-coupon bond maturing at $t=2$. Hence, absent premature surrender, the asset payoff at maturity is
\[
	X_2
	:=
	(E+S_{A,0}+S_{B,0})(1+y_0)^2.
\]
We assume that initial equity is nonnegative and that, absent premature surrender, the insurer can meet the contractual claims at maturity:
\[
	E\geq 0,
	\qquad
	X_2\geq S_{A,2}+S_{B,2}.
\]
At the interim date $t=1$, the market value of these assets is
\[
	A_1:=\frac{X_2}{1+y_1}.
\]

We say that contractual surrender claims are \emph{fully funded} in market value at $t=1$ if
\[
	A_1\geq S_{A,1}+S_{B,1}.
\]
Conversely, the insurer faces \emph{underfunded} contractual claims at $t=1$ if
\[
	A_1<S_{A,1}+S_{B,1}.
\]

At $t=1$, each policyholder $i\in\{A,B\}$ chooses an action
\[
	\alpha_i\in\{\textsf{S},\textsf{C}\},
\]
where $\textsf{S}$ denotes surrender and $\textsf{C}$ denotes continuation.

For $j\neq i$, policyholder $i$'s surrender payout at $t=1$ is denoted by $x_{i,1}(\alpha_j)$ and is defined as
\[
	x_{i,1}(\alpha_j)
	:=
	\begin{cases}
		S_{i,1},
		& \text{if } A_1\geq S_{A,1}+S_{B,1},\\[0.35em]
		\min\{S_{i,1},A_1\},
		& \text{if } A_1<S_{A,1}+S_{B,1}
		\text{ and } \alpha_j=\textsf{C},\\[0.35em]
		\dfrac{S_{i,1}}{S_{i,1}+S_{j,1}}A_1,
		& \text{if } A_1<S_{A,1}+S_{B,1}
		\text{ and } \alpha_j=\textsf{S}.
	\end{cases}
\]
The three cases correspond to full payment, capped individual recovery, and pro-rata recovery under joint surrender.

If policyholder $i$ continues at $t=1$, her monetary payout at maturity is denoted by $x_{i,2}(\alpha_j)$ and is defined as
\[
	x_{i,2}(\alpha_j)
	=
	\begin{cases}
		\min\{S_{i,2},X_2-x_{j,1}(\textsf{C})(1+y_1)\},
		& \text{if } A_1<S_{A,1}+S_{B,1}
		\text{ and } \alpha_j=\textsf{S},\\[0.35em]
		S_{i,2},
		& \text{otherwise}.
	\end{cases}
\]
The term \(x_{j,1}(\textsf{C})(1+y_1)\) is the terminal value of the assets liquidated at \(t=1\) to pay policyholder \(j\)'s surrender claim. In the intermediate region studied below, \(x_{j,1}(\textsf{C})=S_{j,1}\), so the assets remaining for policyholder \(i\) at maturity are \(X_2-S_{j,1}(1+y_1)\).

The payoff comparison at $t=1$ is made in terminal-date units. If policyholder $i$ surrenders at $t=1$, her surrender proceeds can be reinvested for one period at the default-free market rate $y_1$. The terminal value of surrendering is therefore
\[
	x_{i,1}(\alpha_j)(1+y_1).
\]
If policyholder $i$ continues, her monetary payout at $t=2$ is
$x_{i,2}(\alpha_j)$.

\section{Surrender frictions}\label{frictions}

We assume that continuation provides an additional payoff-equivalent benefit that is proportional to the realized maturity payout. For $\psi\geq 0$, the continuation benefit is
\[
    \psi x_{i,2}(\alpha_j).
\]
The total $t=2$ payoff-equivalent value of continuation is then
\[
    x_{i,2}(\alpha_j)(1+\psi).
\]

The benefit $\psi x_{i,2}(\alpha_j)$ is a reduced-form payoff-equivalent surrender friction.
It represents a policyholder's value of maintaining the life-insurance contract beyond the monetary savings payoff.
The proportional specification is exact when the continuation benefits scale with the recoverable policy claim and is a local approximation otherwise.

The analysis focuses on surrender incentives generated by market-value underfunding and strategic interaction. We therefore maintain a lower bound \(\psi\geq\psi'\), derived below, which rules out strict surrender incentives while aggregate surrender claims remain fully funded. Without this restriction, surrender can be privately optimal for all policyholders whenever \(y_1>\lambda y_0\), independently of underfunding; see Remark~\ref{nobenefits}.

Empirical evidence supports the presence of economically relevant surrender frictions. The Bundesbank's 2023 household survey reports that around one quarter of owners of life-insurance policies with guaranteed returns would cancel if a low-risk bank investment yielded at least \(6\%\), while more than half would not surrender regardless of the alternative interest rate. The main stated obstacle was the loss of insurance cover against life risks; other impediments included cancellation fees, upfront fees, effort, and loss of lifelong pension protection \parencites[pp.~45--46]{bundesbank2023fsr}[p.~23]{bundesbank2023bophh44}. The Financial Stability Committee and the Bundesbank's 2025 review report the same qualitative assessment \parencites[p.~20]{afs2024eleventh}{bundesbank2025fsr}.

The parameter \(\psi\) captures the payoff-dependent component of continuation value: the payoff-equivalent of retained contract features whose value varies with the recoverable continuation payout. The most relevant examples are biometric and supplementary risk cover, such as death, accident, invalidity, or occupational-disability cover, and pension or annuitization protection. Embedded guarantees, bonus rights, tax advantages, and subsidy benefits are captured by this parameter to the extent that their value is tied to the recoverable policy claim.

Fixed surrender costs, such as information and transaction costs, are not payoff-dependent in general. They are represented by $\psi$ only as local proportional equivalents around the relevant surrender comparison. They shift individual exercise thresholds, but are not by themselves a source of strategic complementarity.
Retained contractual cancellation fees are distinct because they increase cash flows to the insurer as well as the residual asset pool. They are therefore excluded from the analysis.

\section{A strategic surrender game}\label{runs}

We now consider a static pure-strategy surrender game between two policyholders. There are two contracts with premiums $S_{A,0}$ and $S_{B,0}$ and a single decision date $t=1$. A pair
\[
	(\alpha_A,\alpha_B)\in\{\textsf{S},\textsf{C}\}^2
\]
is called a \emph{strategy profile}.

For $i\in\{A,B\}$ and $j\neq i$, define policyholder $i$'s terminal surrender value at $t=1$, conditional on $j$'s action, by
\begin{align}
	V_i^S(\alpha_j)
	:=
	x_{i,1}(\alpha_j)(1+y_1),
	\label{eq:VS-def}
\end{align}
where the surrender proceeds are reinvested for one period at the market rate $y_1$. Define the continuation value by
\begin{align}
	V_i^C(\alpha_j)
	:=
	x_{i,2}(\alpha_j)(1+\psi).
	\label{eq:VC-def}
\end{align}
The factor $1+\psi$ is the gross payoff-dependent continuation benefit introduced in Section~\ref{frictions}. All payoffs in the surrender game are expressed in terminal-date units.

We adopt a tie-breaking convention: whenever a policyholder is indifferent between surrender and continuation, she continues. Hence policyholder $i$'s best-response rule is
\begin{align}
	\text{BR}_i(\alpha_j)
	:=
	\begin{cases}
		\textsf{S}, & \text{if } V_i^S(\alpha_j)>V_i^C(\alpha_j),\\
		\textsf{C}, & \text{if } V_i^S(\alpha_j)\leq V_i^C(\alpha_j).
	\end{cases}
	\label{eq:BR-def}
\end{align}
Equivalently, $i$ surrenders if and only if
\begin{align}
	x_{i,1}(\alpha_j)(1+y_1)>x_{i,2}(\alpha_j)(1+\psi).
	\label{rulepsi}
\end{align}
A strategy profile $(\alpha_A,\alpha_B)$ is a \emph{Nash equilibrium} under the continuation tie-break rule if
\begin{align}
	\alpha_A=\text{BR}_A(\alpha_B)
	\quad\text{and}\quad
	\alpha_B=\text{BR}_B(\alpha_A).
	\label{eq:NE-def}
\end{align}

Our objective is to separate parameter regions with a unique equilibrium from those with multiple equilibria. When the equilibrium is unique and equal to $(\textsf{C},\textsf{C})$, surrender runs are precluded. When the equilibrium is unique and equal to $(\textsf{S},\textsf{S})$, surrender is fundamentals-driven. When both $(\textsf{C},\textsf{C})$ and $(\textsf{S},\textsf{S})$ are equilibria at the same interest rate, surrender decisions exhibit strategic complementarity and the life insurer is vulnerable to self-fulfilling runs.

The normal form of the game can be summarized as follows. In each cell, the first component is policyholder $A$'s payoff and the second component is policyholder $B$'s payoff, both measured in terminal-date units:
\begin{align*}
\begin{array}{c|cc}
	& \alpha_B=\textsf{C} & \alpha_B=\textsf{S} \\ \hline
	\alpha_A=\textsf{C}
	&
	\left(x_{A,2}(\textsf{C})(1+\psi),\;x_{B,2}(\textsf{C})(1+\psi)\right)
	&
	\left(x_{A,2}(\textsf{S})(1+\psi),\;x_{B,1}(\textsf{C})(1+y_1)\right)
	\\[0.35em]
	\alpha_A=\textsf{S}
	&
	\left(x_{A,1}(\textsf{C})(1+y_1),\;x_{B,2}(\textsf{S})(1+\psi)\right)
	&
	\left(x_{A,1}(\textsf{S})(1+y_1),\;x_{B,1}(\textsf{S})(1+y_1)\right).
\end{array}
\end{align*}
The table identifies the strategic channel: one policyholder's surrender can reduce the other policyholder's realized continuation payoff and, because the continuation benefit is proportional to that payoff, also the value of continuation.

The interest-rate threshold at which aggregate surrender claims are exactly funded in market value is
\begin{align}
	y^*
	:=
	\frac{X_2}{S_{A,1}+S_{B,1}}-1.
	\label{eq:ystar-game}
\end{align}
At this threshold, the insurer can exactly pay both surrender claims from asset proceeds. Above it, aggregate surrender claims exceed the market value of assets at $t=1$.

Define
\begin{align}
	y^{\max}
	:=
	\min\left\{
		\frac{X_2}{S_{A,1}}-1,\;
		\frac{X_2}{S_{B,1}}-1
	\right\}.
	\label{eq:ymax}
\end{align}
We focus on the intermediate underfunding region in which the insurer can cover one surrender claim but not both simultaneously:
\begin{align}
	y^*<y_1\leq y^{\max}.
	\label{game.eq.1}
\end{align}
In this region, one policyholder’s surrender reduces the assets backing the other policyholder’s continuation claim.

The threshold
\begin{align}
	\hat y:=(1+\psi)(1+\lambda y_0)-1
	\label{eq:haty-small}
\end{align}
is the interest rate above which surrender becomes a strict best response to the other policyholder's continuation. Indeed, if $j$ continues and the insurer can cover one surrender claim, then
\[
	V_i^S(\textsf{C})=S_{i,1}(1+y_1),
	\qquad
	V_i^C(\textsf{C})=S_{i,2}(1+\psi),
\]
so $V_i^S(\textsf{C})>V_i^C(\textsf{C})$ is equivalent to $y_1>\hat y$.

To isolate surrender runs generated under market-value underfunding rather than rate-driven surrender under full funding, we impose the no-surrender condition under full funding. This condition is equivalent to the lower bound
\begin{align}
	\psi'
	:=
	\frac{1+y^*}{1+\lambda y_0}-1
	=
	\frac{X_2}{S_{A,2}+S_{B,2}}-1
	\label{eq:psi-prime}
\end{align}
on the continuation-benefit parameter. Indeed,
\begin{align}
	\hat y-y^*
	=
	(1+\lambda y_0)(\psi-\psi').
	\label{eq:hat-minus-ystar}
\end{align}
Because \(1+\lambda y_0>0\), continuation is selected throughout the full-funding region \(y_1\leq y^*\) if and only if \(\psi\geq\psi'\). For the rest of this section and the existence results below, we maintain \(\psi\geq\psi'\), with strict inequalities interpreted using the tie-breaking convention in \eqref{eq:BR-def}.

At the lower bound,
\(
    \hat y=y^*.
\)
For \(\psi>\psi'\), the surrender-against-continuation threshold lies strictly above the underfunding threshold.

\begin{remark}\label{nobenefits}
If \(\psi=0\), then \(\hat y(0)=\lambda y_0\). Maturity solvency implies
\[
    X_2\geq (S_{A,0}+S_{B,0})(1+\lambda y_0)^2,
\]
and therefore
\[
    1+y^*
    =
    \frac{X_2}{(S_{A,0}+S_{B,0})(1+\lambda y_0)}
    \geq
    1+\lambda y_0.
\]
Hence \(y^*\geq \lambda y_0\), with strict inequality under strict maturity solvency. In the full-funding region, the payoff comparison is
\[
    V_i^S-V_i^C
    =
    S_{i,0}(1+\lambda y_0)(y_1-\lambda y_0),
\]
so surrender is strictly optimal whenever \(y_1>\lambda y_0\). Thus, under strict maturity solvency, there is a nonempty full-funding interval \((\lambda y_0,y^*]\) on which joint surrender is already the unique Nash equilibrium under the continuation tie-break rule. Without continuation benefits, \(y^*\) is therefore not the first rate at which surrender becomes privately optimal, and the model cannot isolate surrender runs generated by market-value underfunding and strategic interaction. The same comparison also implies that surrender is a strict best response to surrender throughout the intermediate region. This motivates the lower bound \(\psi\geq\psi'\).
\end{remark}

Next, define, for $i\in\{A,B\}$ and $j\neq i$,
\begin{align}
	\tilde y_i
	:=
	\frac{X_2-S_{i,2}}{S_{j,1}}-1.
	\label{eq:tilde-y}
\end{align}
\(\tilde y_i\) is the regime-switch threshold for policyholder \(i\)'s continuation payoff after policyholder \(j\) surrenders. For \(y_1\leq \tilde y_i\), the remaining terminal assets after \(j\)'s surrender are still sufficient to pay \(i\)'s full contractual maturity claim, so \(x_{i,2}(\textsf{S})=S_{i,2}\). For \(y_1>\tilde y_i\), the remaining terminal assets are insufficient, so \(i\)'s continuation payoff is impaired and equals \(X_2-S_{j,1}(1+y_1)\).

Finally, define
\begin{align}
	\bar y_i
	:=
	\frac{X_2}{S_{j,1}}
	\left(
		1-\frac{S_{i,0}}{(S_{A,0}+S_{B,0})(1+\psi)}
	\right)-1.
	\label{eq:bary-i}
\end{align}
The threshold $\bar y_i$ is obtained from the comparison between the pro-rata surrender payoff and the impaired continuation payoff when the remaining asset value is below $S_{i,2}$. If \(\psi\geq\psi'\), then
\begin{align}
	\bar y_i-\tilde y_i=
	\frac{S_{i,0}(1+\lambda y_0)}{S_{j,0}}
	\frac{\psi-\psi'}{1+\psi}
	\geq 0.
	\label{eq:bary-minus-tilde}
\end{align}
Under \(\psi\geq\psi'\), the selected surrender-against-surrender threshold is \(\bar y_i\).

\begin{lemma}[best responses]\label{lemma1}
	Let \(\psi\geq\psi'\). Within the intermediate region \eqref{game.eq.1},
	\begin{itemize}
		\item $\text{\upshape BR}_i(\textsf{\upshape C})=\textsf{\upshape S}$ if and only if
		\(
			y_1>\hat y;
		\)
		\item $\text{\upshape BR}_i(\textsf{\upshape S})=\textsf{\upshape S}$ if and only if
		\(
			y_1>\bar y_i.
		\)
	\end{itemize}
\end{lemma}

\begin{proof}
Appendix \ref{sec:proof_lemma1}.
\end{proof}

Define the joint surrender-against-surrender threshold and the strict joint-surrender cutoff by
\[
        \bar y:=\max\{\bar y_A,\bar y_B\},
        \qquad
       T_{\textsf{SS}}:=\max\{\hat y,\bar y\}.
\]
The interval
\begin{align}
	\hat Y
	:=
	\left(\bar y,\,\hat y\right]
	\cap
	\left(y^*,y^{\max}\right]
	\label{eq:coordination-region}
\end{align}
is the coordination region under the continuation tie-break rule. In this interval, continuation is the selected best response to continuation, while surrender is the selected best response to surrender. The interval
\begin{align}
	\overline Y
	:=
	\left(T_{\textsf{SS}},\,y^{\max}\right]
	\label{eq:fundamental-region}
\end{align}
is the fundamentals-driven run region within the intermediate region. In this interval, surrender is the selected best response to both continuation and surrender for both policyholders.

The following proposition characterizes the pure-strategy equilibria in the intermediate region.

\begin{proposition}[pure-strategy equilibria]\label{proposition2}
	For \(\psi\geq\psi'\) the pure-strategy Nash equilibria within the intermediate region \eqref{game.eq.1} are characterized as follows:	
	\begin{itemize}
		\item \((\textsf{\upshape C},\textsf{\upshape C})\) is a Nash equilibrium if and only if \(y_1\leq \hat y\);
		\item \((\textsf{\upshape S},\textsf{\upshape S})\) is a Nash equilibrium if and only if \(y_1>\bar y\);
		\item \((\textsf{\upshape S},\textsf{\upshape C})\) is a Nash equilibrium if and only if 
\(y_1>\hat y\) and \(y_1\leq \bar y_B\);
		\item \((\textsf{\upshape C},\textsf{\upshape S})\) is a Nash equilibrium if and only if
\(y_1>\hat y\) and \(y_1\leq \bar y_A\).
	\end{itemize}
\end{proposition}

\begin{proof}
Appendix \ref{sec:proof_prop2}.
\end{proof}

Proposition~\ref{proposition2} classifies equilibria for a given continuation-benefit parameter \(\psi\).\footnote{The off-diagonal equilibria represent strategic substitutability. Contract-size heterogeneity can make exit incentives asymmetric and select between the two off-diagonal profiles.} The next two theorems show when the corresponding run regions are nonempty for some admissible value of \(\psi\).

Define
\[
    \psi^U:=\frac{1+y^{\max}}{1+\lambda y_0}-1.
\]
\(\psi^U\) is the largest continuation-benefit parameter for which the best-response threshold for surrender against continuation has not yet exited the intermediate underfunding region, \(\hat y(\psi^U)=y^{\max}\). 

We now state our main results. Theorem \ref{theorem1} gives an existence condition for self-fulfilling runs. Theorem \ref{theorem2} gives an existence condition for fundamentals-driven runs.
\begin{theorem}[self-fulfilling runs]\label{theorem1}
There exist \(\psi\geq\psi'\) and \(y_1\) in the intermediate region \eqref{game.eq.1} such that
\(
    (\emph{\textsf{\upshape C}},\emph{\textsf{\upshape C}})
\)
and
\(
    (\emph{\textsf{\upshape S}},\emph{\textsf{\upshape S}})
\)
are both Nash equilibria under the continuation tie-break rule at \(y_1\) if and only if \(\bar y(\psi^U)<y^{\max}.\)
\end{theorem}

\begin{proof}
Appendix \ref{sec:proof_theorem1}.
\end{proof}
Payoff-dependent continuation benefits can make the best-response thresholds cross, generating a coordination region in which both joint continuation and joint surrender are Nash equilibria. The relevant boundary is \(\psi^U\), where the surrender-against-continuation threshold reaches the upper end of the intermediate region.

Let
\[
   \tilde y:=\max\left\{\tilde y_A,\tilde y_B\right\}.
\]
\begin{theorem}[fundamentals-driven runs]\label{theorem2}
There exist \(\psi\geq\psi'\) and \(y_1\) in the intermediate region \eqref{game.eq.1} such that $(\textsf{\upshape S},\textsf{\upshape S})$ is the unique Nash equilibrium under the continuation tie-break rule at \(y_1\) if and only if \(\tilde y<y^{\max}.\)
\end{theorem}

\begin{proof}
Appendix \ref{sec:proof_theorem2}.
\end{proof}
A fundamentals-driven run exists when joint surrender is the unique Nash equilibrium for some admissible continuation benefit and interim rate.\footnote{The condition in Theorem~\ref{theorem2} is weaker than the condition in Theorem~\ref{theorem1}. Indeed, since $\psi^U>\psi'$ and $\bar y$ is strictly increasing in $\psi$, while $\bar y_i(\psi')=\tilde y_i$ for $i\in\{A,B\}$,
\(
    \tilde y
    =\bar y(\psi')
    <\bar y(\psi^U)
    <y^{\max}.
\)}
Since the relevant surrender thresholds are increasing in \(\psi\), the most favorable case for fundamentals-driven surrender is the smallest admissible continuation benefit \(\psi'\).

\begin{figure}[H]
\centering
\begin{tikzpicture}[x=1.05cm,y=1cm]


\node[anchor=west,font=\small] at (0,1.15)
    {\textbf{(a) Coordination region:}
    $\bar y<\hat y$};

\draw[->] (0,0) -- (12,0) node[right] {$y_1$};

\draw (1,0.12) -- (1,-0.12)
    node[above=4pt] {$y^*$};

\draw (4.3333,0.12) -- (4.3333,-0.12)
    node[above=4pt] {$\bar y$};

\draw (7.6667,0.12) -- (7.6667,-0.12)
    node[above=4pt] {$\hat y$};

\draw (11,0.12) -- (11,-0.12)
    node[above=4pt] {$y^{\max}$};

\draw[
    decorate,
    decoration={brace,mirror,amplitude=5pt}
]
    (1,-0.38) -- (4.3333,-0.38)
    node[midway,below=8pt,font=\scriptsize,align=center]
    {\((\textsf{C},\textsf{C})\)};

\draw[
    decorate,
    decoration={brace,mirror,amplitude=5pt}
]
    (4.3333,-0.38) -- (7.6667,-0.38)
    node[midway,below=8pt,font=\scriptsize,align=center]
    {\((\textsf{C},\textsf{C})\), \((\textsf{S},\textsf{S})\)};

\draw[
    decorate,
    decoration={brace,mirror,amplitude=5pt}
]
    (7.6667,-0.38) -- (11,-0.38)
    node[midway,below=8pt,font=\scriptsize,align=center]
    {\((\textsf{S},\textsf{S})\)};


\node[anchor=west,font=\small] at (0,-1.80)
    {\textbf{(b) No coordination region:}
    $\hat y \leq \bar y$};

\draw[->] (0,-2.95) -- (12,-2.95) node[right] {$y_1$};

\draw (1,-2.83) -- (1,-3.07)
    node[above=4pt] {$y^*$};

\draw (4.3333,-2.83) -- (4.3333,-3.07)
    node[above=4pt] {$\hat y$};

\draw (7.6667,-2.83) -- (7.6667,-3.07)
    node[above=4pt] {$\bar y$};

\draw (11,-2.83) -- (11,-3.07)
    node[above=4pt] {$y^{\max}$};

\draw[
    decorate,
    decoration={brace,mirror,amplitude=5pt}
]
    (1,-3.33) -- (4.3333,-3.33)
    node[midway,below=8pt,font=\scriptsize,align=center]
    {\((\textsf{C},\textsf{C})\)};

\draw[
    decorate,
    decoration={brace,mirror,amplitude=5pt}
]
    (4.3333,-3.33) -- (7.6667,-3.33)
    node[midway,below=8pt,font=\scriptsize,align=center]
    {\((\textsf{C},\textsf{S})\), \((\textsf{S},\textsf{C})\)};

\draw[
    decorate,
    decoration={brace,mirror,amplitude=5pt}
]
    (7.6667,-3.33) -- (11,-3.33)
    node[midway,below=8pt,font=\scriptsize,align=center]
    {\((\textsf{S},\textsf{S})\)};

\end{tikzpicture}
\caption{Schematic ordering of interest-rate thresholds for fixed $\psi$ and equal contract sizes. Panel (a) depicts a nonempty coordination region. Panel (b) depicts the case in which the coordination region is empty, but fundamentals-driven runs remain possible.}
\label{fig:threshold-ordering}
\end{figure}

\section{Comparative statics}\label{compstat}

This section presents the comparative statics of the two best-response thresholds and of the joint-surrender cutoff. Proposition~\ref{prop:comp_stat_surrender_cutoff} separates the return-comparison threshold \(\hat y\) from the strategic threshold \(\bar y\). The former depends only on credited contractual growth and continuation benefits; the latter also depends on the own-funds buffer and policy-size distribution because it applies after the other policyholder surrenders.

\vspace{2mm}

Let \(s:=S_{A,0}+S_{B,0}\), \(L:=\max\{S_{A,0},S_{B,0}\}\), and \(\theta:=L/s\).

\begin{proposition}[comparative statics]\label{prop:comp_stat_surrender_cutoff}

Under the maintained assumptions, and for parameter variations that preserve \(\psi\geq\psi'\), the following hold:
\begin{enumerate}
	\item[\textup{(i)}] \(\hat y\) is strictly increasing in \(\psi\) and \(y_0\). It is strictly increasing, constant, or strictly decreasing in \(\lambda\) according as \(y_0>0\), \(y_0=0\), or \(y_0<0\).
	\item[\textup{(ii)}] \(\bar y\) is strictly increasing in \(E\), \(\psi\), and \(y_0\). It is strictly decreasing, constant, or strictly increasing in \(\lambda\) according as \(y_0>0\), \(y_0=0\), or \(y_0<0\).
	\item[\textup{(iii)}] \(T_{\emph{\textsf{SS}}}\) is weakly increasing in \(E\) and strictly increasing in \(\psi\) and \(y_0\). For \(y_0\neq0\), \(T_{\emph{\textsf{SS}}}\) has no unconditional monotone comparative static in \(\lambda\).
	\item[\textup{(iv)}]  Holding \(s\) fixed, \(\bar y\) is weakly increasing in \(\theta\), strictly if \(\psi>0\). Holding \(\theta\) and \(E\) fixed, \(\bar y\) is weakly decreasing in \(s\), strictly if \(E>0\).
\end{enumerate}
\end{proposition}

\begin{proof}
Appendix \ref{sec:proof_prop3}.
\end{proof}

The effect of \(\lambda\) is branch-dependent. For \(y_0>0\), a higher \(\lambda\) raises \(\hat y\) but lowers \(\bar y\) by increasing the other policyholder's interim surrender payment; for \(y_0<0\), the signs reverse. The selected cutoff \(T_{\textsf{SS}}\) therefore inherits the local sign of the binding branch. Greater policy-size asymmetry raises \(\bar y\), except at \(\psi=0\), because the larger policy determines the binding surrender-against-surrender threshold.

\section{Parametrization}\label{parametrisation}

This section maps the objects from Proposition \ref{prop:comp_stat_surrender_cutoff} into economically interpretable threshold levels under the normalization \(S_{A,0}=S_{B,0}=1\). It first links the model parameters to empirical magnitudes and then illustrates how capitalization, profit participation, and continuation benefits shape strict joint surrender and coordination.

\subsection{Equity}

We interpret \(E\) as a stylized own-funds buffer. The regulatory analog is the Solvency II balance sheet surplus, i.e. the excess of assets over liabilities on the market-consistent Solvency II valuation basis \parencite[Arts.~75 and 88]{SolvencyIIDirective}.\footnote{Solvency II capital adequacy is defined as eligible own funds relative to the Solvency Capital Requirement, set at one-year 99.5\% Value-at-Risk of basic own funds \parencite[Art.~101]{SolvencyIIDirective}. We abstract from regulatory features with no direct counterpart in the model, including subordinated liabilities and ancillary own funds \parencite[Arts.~88--89]{SolvencyIIDirective} and own-funds tiering and eligibility restrictions \parencite[Art.~93]{SolvencyIIDirective}; see also \textcite[Art.~82]{DelegatedRegulation2015_35} for quantitative tier eligibility limits.} 

Under the normalization \(S_{A,0}=S_{B,0}=1\), the ratio
\[
    e:=\frac{E}{E+2}
\]
is an asset-based equity share, while the corresponding equity-over-liabilities ratio is
\[
    \frac{E}{2}=\frac{e}{1-e}.
\]
Thus \(e=5\%\) corresponds to \(E/2\simeq5.3\%\), and \(e=10\%\) corresponds to \(E/2\simeq11.1\%\), or \(A/L\simeq111.1\%\).
As an external benchmark, EIOPA's 2024 insurance stress test baseline reports an aggregate assets-over-liabilities ratio of \(111.3\%\) and excess assets over liabilities of EUR \(656\) billion, implying \(e=1-1/1.113\simeq10.15\%\) in our normalization
\parencite{EIOPA2024StressTest}.
We therefore use the numerical interval \(e\in[0,10\%]\). 

\subsection{Base interest rate and profit participation}

The baseline rate \(y_0\) is interpreted as the annualized return on the insurer's backing asset pool, or equivalently as a low-risk term-structure benchmark for the insurer's pre-shock asset payoff. The parametrization uses a base rate of \(y_0=2\%\). The profit-participation parameter \(\lambda\) maps the gross return entering the contract into the credited contractual growth rate \(\lambda y_0\).

EIOPA's Costs and Past Performance Report provides broad European evidence on asset returns and profit participation for insurance-based investment products: its 2026 survey covers undertakings representing more than \(60\%\) of EEA unit-linked and profit participation premium, and reports annualized and product-level net returns for profit participation insurance-based investment products, net of costs and weighted by gross written premiums \parencite[p.~3; pp.~12--15, figs.~8 and~10--13]{EIOPACPP2026}. 

According to IVASS, Italian \emph{gestioni separate} had an average gross return of \(2.5\%\) in 2021, an average return retroceded to policyholders of \(1.4\%\), and an average retained return of \(1.1\%\) \parencite{IVASSGese2021}. In 2024, Italian separately managed accounts without a profit fund had an average gross return of \(2.8\%\) and an average return retroceded to policyholders of \(1.7\%\). The same report mentions \(3.6\%\) gross and \(2.1\%\) retroceded for the separately managed accounts with a profit fund \parencite[pp.~11--12, fig.~I.8]{IVASSAnnual2024}. These figures imply pass-through ratios of \(0.56\), \(0.607\), and \(0.583\).

French evidence points to a similar order of magnitude. ACPR reports a 2021 average revaluation rate of \(1.28\%\) for individual euro-denominated life-insurance contracts and an asset-return measure of \(2.2\%\), implying a pass-through ratio of \(0.58\) \parencite{ACPRRevalorisation2021}. For 2024, ACPR reports an average revaluation rate of \(2.63\%\) and an average asset return of \(2.5\%\) \parencite[pp.~4, 7, and~14]{ACPRRevalorisation2024}. An adjustment for bonuses of 40 basis points yields a ratio of \(0.89\) \parencite[p. 7]{ACPRRevalorisation2024}.

German law provides a regulatory benchmark for the pass-through ratio: the MindZV minimum-allocation rule ties the allocation from eligible investment income for surplus-entitled life-insurance contracts to \(90\%\) of the relevant investment result, subject to the statutory formula and guaranteed-interest deductions \parencite[\S~6(1)]{MindZV2016}. We therefore take \(\lambda\in[0.5,0.9]\) for the comparative statics grid. The lower end is consistent with the Italian evidence on retained returns and allows for costs, guarantees, and reserve building; the upper end is the benchmark implied by German law.

\subsection{Interim interest rate}

The interim interest rate \(y_1\) of the two-period model corresponds to a one-period yield equivalent of the market-value shock to the assets funding surrender payments. In the model,
\[
    A_1(y_1)
    =
    \frac{X_2}{1+y_1},
    \qquad
    A_1(y_0)
    =
    \frac{X_2}{1+y_0}.
\]
Hence a portfolio-value loss \(h\), measured relative to the ex ante $t=1$ asset value, satisfies
\[
    1-h
    =
    \frac{A_1(y_1)}{A_1(y_0)}
    =
    \frac{1+y_0}{1+y_1}.
\]
The one-period shock equivalent is therefore
\begin{equation}
    y_1
    =
    \frac{1+y_0}{1-h}-1.
    \label{eq:y1-effective}
\end{equation}

Equation~\eqref{eq:y1-effective} can also be read in the opposite
direction. If a figure reports an effective rate shock
\[
    \Delta y:=y_1-y_0,
\]
then the corresponding market-value loss is
\[
    h(\Delta y;y_0)
    =
    1-\frac{1+y_0}{1+y_0+\Delta y}.
\]
The numerical example below (Table \ref{tab:tss-asset-loss-thresholds}) produces threshold interest rates between \(3.8\%\) and \(26.6\%\) at \(\lambda=0.7\), corresponding to asset-value shocks between \(1.7\%\) and \(19.4\%\).

\subsection{Numerical analysis}

Let
\[
    q:=1+\lambda y_0,
    \qquad
    a(e,\lambda):=
    \frac{2(1+y_0)^2}{(1-e)q^2}.
\]
Then the coordination-onset benefit can be written as
\[
    \psi^\times(e,\lambda)
    :=
    \frac{a(e,\lambda)+
    \sqrt{a(e,\lambda)^2-2a(e,\lambda)}}{2}-1.
\]
The corresponding coordination-onset interest rate is
\[
    y^\times(e,\lambda)
    :=
    \hat y\left(\psi^\times(e,\lambda)\right)
    =
    \bar y\left(\psi^\times(e,\lambda);e\right).
\]

Write \(\bar y(\psi):=\bar y_A(\psi)=\bar y_B(\psi)\). The coordination-onset boundary is defined by
\[
    \hat y(\psi^\times)=\bar y(\psi^\times;e).
\]

The following table summarizes the parameter ranges used as inputs for the computations and figures below. 
\begin{table}[H]
\centering
\caption{Baseline design and admissible ranges}
\label{tab:equal-contract-compstat-design}
\small
\begin{tabular*}{\textwidth}{@{\extracolsep{\fill}}p{0.39\textwidth}p{0.55\textwidth}@{}}
\toprule
Object & Value or range \\
\midrule
Contract sizes & \(S_{A,0}=S_{B,0}=1\) \\
Base interest rate & \(y_0=2\%\) \\
Profit participation & \(\lambda\in[0.5,0.9]\) \\
Own-funds & \(e\in[0,10\%]\) \\
Interest-rate shocks (Figure 3) & \(y_1-y_0\in[0,25]\) pp \\
Continuation-benefit & \(\psi\in[\psi'(e,\lambda),100\%]\) \\
Full-funding lower bound & \(\psi'=X_2/[2(1+\lambda y_0)^2]-1\) \\
\bottomrule
\end{tabular*}

\vspace{0.35em}
\begin{minipage}{\textwidth}
\footnotesize
\textbf{Notes:} The upper intermediate-region bound does not bind over this grid: for \(e\in[0,10\%]\) and \(\lambda\in[0.5,0.9]\), \(\psi^U(e,\lambda)\ge100.787\%\). Thus the range \([\psi'(e,\lambda),100\%]\) lies inside the intermediate-region domain.
\end{minipage}
\end{table}
The next table reports threshold values by own-funds ratio and profit participation. The value \(\psi^\times\) is the coordination boundary, \(y^\times\) is the corresponding rate.
\begin{table}[H]
\centering
\caption{Continuation-benefit and interest-rate thresholds as a function of own funds}
\label{tab:coordination-onset-equal-contracts}
\small
\begin{tabular*}{\textwidth}{@{\extracolsep{\fill}}lcccc@{}}
\toprule
Own-funds ratio \(e\)
& \(\psi'\)
& \(\psi^\times\)
& \(y^\times\)
& \(T_{\mathsf{SS}}(\psi')-y_0\) \\
\midrule
\multicolumn{5}{@{}l}{\(\lambda=0.5\)} \\
\(0\%\)  & \(1.990\%\)  & \(16.236\%\) & \(17.399\%\) & \(3.020\) pp \\
\(2\%\)  & \(4.071\%\)  & \(24.656\%\) & \(25.902\%\) & \(7.224\) pp \\
\(5\%\)  & \(7.358\%\)  & \(35.464\%\) & \(36.818\%\) & \(13.863\) pp \\
\(10\%\) & \(13.322\%\) & \(52.177\%\) & \(53.699\%\) & \(25.911\) pp \\
\midrule
\multicolumn{5}{@{}l}{\(\lambda=0.7\)} \\
\(0\%\)  & \(1.187\%\)  & \(12.146\%\) & \(13.716\%\) & \(1.807\) pp \\
\(2\%\)  & \(3.252\%\)  & \(21.576\%\) & \(23.278\%\) & \(5.995\) pp \\
\(5\%\)  & \(6.513\%\)  & \(32.850\%\) & \(34.710\%\) & \(12.607\) pp \\
\(10\%\) & \(12.430\%\) & \(49.813\%\) & \(51.910\%\) & \(24.608\) pp \\
\midrule
\multicolumn{5}{@{}l}{\(\lambda=0.9\)} \\
\(0\%\)  & \(0.393\%\)  & \(6.677\%\)  & \(8.597\%\)  & \(0.601\) pp \\
\(2\%\)  & \(2.442\%\)  & \(18.259\%\) & \(20.388\%\) & \(4.772\) pp \\
\(5\%\)  & \(5.677\%\)  & \(30.171\%\) & \(32.514\%\) & \(11.359\) pp \\
\(10\%\) & \(11.548\%\) & \(47.439\%\) & \(50.093\%\) & \(23.312\) pp \\
\bottomrule
\end{tabular*}

\vspace{0.35em}
\begin{minipage}{\textwidth}
\footnotesize
\textbf{Notes:} The table reports composite objects that depend on \(\lambda\) through \(\psi'(e,\lambda)\) and \(\psi^\times(e,\lambda)\). For \(y_0>0\), increasing \(\lambda\) raises the \(\hat y\)-branch and lowers the \(\bar y\)-branch. Hence the sign of the local effect on \(T_{\mathsf{SS}}\) depends on which branch binds.
\end{minipage}
\end{table}

The final column of Table~\ref{tab:coordination-onset-equal-contracts} shows strict joint-surrender cutoffs as shocks above \(y_0\). These thresholds can also be expressed in asset-loss units. For any cutoff \(T_{\mathsf{SS}}\), define
\[
    h_{\mathsf{SS}}
    :=
    1-\frac{A_1(T_{\mathsf{SS}})}{A_1(y_0)}
    =
    1-\frac{1+y_0}{1+T_{\mathsf{SS}}}.
\]
Because \(h(y_1;y_0)\) is strictly increasing in \(y_1\), the condition \(y_1>T_{\mathsf{SS}}\) is equivalent to \(h(y_1;y_0)>h_{\mathsf{SS}}\). The following table applies this conversion to the baseline \(\lambda=0.7\) rows of Table~\ref{tab:coordination-onset-equal-contracts}.

Table~\ref{tab:tss-asset-loss-thresholds} expresses the model cutoffs in asset-loss units, allowing comparison with observed insurer portfolio losses. Italy in 2022--2023 provides a useful benchmark. IVASS reports that the rise in interest rates in 2022 generated a negative balance of capital gains and losses of EUR~51.6 billion, equal to \(9.8\%\) of the book value of insurers' securities portfolios excluding linked securities \parencite[p.~7]{IVASSAnnual2022}. It also reports market-value declines of \(16.4\%\) for investments net of unit-linked contracts, \(20.3\%\) for bond investments, and \(24.8\%\) for Italian government bonds \parencite[p.~7; pp.~46--47]{IVASSAnnual2022}. These losses were followed by a substantial increase in surrenders in 2023: IVASS reports that surrenders rose by EUR~32 billion, or \(59\%\), relative to 2022 \parencite{IVASSQuaderno31}. The asset-loss thresholds implied by the model are thus of the same order of magnitude as the reported Italian losses during a period of surrender stress.\footnote{The same period also saw the failure of the Italian life insurer Eurovita, whose crisis followed the 2022 rise in interest rates and spreads, weak solvency, large unrealized fixed-income losses, and rapid surrenders before IVASS suspended surrender rights on 6 February 2023 \parencite{BankOfItalyFSR2023,IVASSAnnual2022,IVASSOrderEurovita2023}.}

\begin{table}[H]
\centering
\caption{Strict joint-surrender cutoffs in asset-loss units for \(\lambda=0.7\)}
\label{tab:tss-asset-loss-thresholds}
\small
\begin{tabular*}{\textwidth}{@{\extracolsep{\fill}}lccc@{}}
\toprule
Own-funds ratio \(e\)
& \(T_{\mathsf{SS}}(\psi')-y_0\)
& \(T_{\mathsf{SS}}(\psi')\)
& \(h_{\mathsf{SS}}\) \\
\midrule
\(0\%\)  & \(1.807\) pp  & \(3.807\%\)  & \(1.741\%\) \\
\(2\%\)  & \(5.995\) pp  & \(7.995\%\)  & \(5.551\%\) \\
\(5\%\)  & \(12.607\) pp & \(14.607\%\) & \(11.001\%\) \\
\(10\%\) & \(24.608\) pp & \(26.608\%\) & \(19.436\%\) \\
\bottomrule
\end{tabular*}

\vspace{0.35em}
\begin{minipage}{\textwidth}
\footnotesize
\textbf{Notes:} The table translates the strict joint-surrender cutoff at the full-funding lower bound \(\psi'\) into the equivalent asset-value loss threshold \(h_{\mathsf{SS}}=1-(1+y_0)/(1+T_{\mathsf{SS}})\).
\end{minipage}
\end{table}

Table \ref{tab:fixed-psi-lambda-compstat} illustrates the \(\lambda\)-comparative static in Proposition~\ref{prop:comp_stat_surrender_cutoff}.

\begin{table}[H]
\centering
\caption{Branch-dependent \(\lambda\)-comparative static}
\label{tab:fixed-psi-lambda-compstat}
\small
\begin{tabular*}{\textwidth}{@{\extracolsep{\fill}}lcccc@{}}
\toprule
\(\lambda\)
& \(T_{\mathsf{SS}}(5\%)\)
& active branch
& \(T_{\mathsf{SS}}(50\%)\)
& active branch \\
\midrule
\(0.5\) & \(10.117\%\) & \(\bar y\) & \(51.5\%\) & \(\hat y\) \\
\(0.7\) & \(9.683\%\)  & \(\bar y\) & \(52.1\%\) & \(\hat y\) \\
\(0.9\) & \(9.252\%\)  & \(\bar y\) & \(52.7\%\) & \(\hat y\) \\
\bottomrule
\end{tabular*}

\vspace{0.35em}
\begin{minipage}{\textwidth}
\footnotesize
\textbf{Notes:} The table fixes \(e=2\%\). At \(\psi=5\%\), the strategic branch binds and \(T_{\mathsf{SS}}\) decreases with \(\lambda\). At \(\psi=50\%\), the return-comparison branch binds and \(T_{\mathsf{SS}}\) increases with \(\lambda\).
\end{minipage}
\end{table}

The following figures illustrate the comparative statics of the key model objects for the parameters in Table \ref{tab:equal-contract-compstat-design}, with \(y_0=2\%\) and \(\lambda=0.7\) fixed.

\begin{figure}[H]
\centering
\begin{tikzpicture}
\begin{axis}[
    thresholdaxis,
    width=0.94\textwidth,
    height=0.45\textwidth,
    grid=none,
    xmin=0, xmax=70,
    ymin=0, ymax=70,
    ytick={0,20,40,60},
    xlabel={continuation-benefit parameter \(\psi\) (\%)},
    ylabel={strict joint-surrender cutoff (\%)},
    legend columns=2,
    legend style={
        at={(0.03,0.97)},
        anchor=north west,
        font=\footnotesize,
        draw=none,
        fill=white,
        fill opacity=0.9,
        text opacity=1
    },
]

\addplot[gray!70!black, solid, line width=0.7pt, domain=0:70, samples=220]
    {100*hatf(0.02,0.7,x/100)};
\addlegendentry{common \(\hat y\)-branch}

\addplot[black, dashdotted, line width=0.75pt, domain=3.252:21.576, samples=160]
    {100*barf(0.02,0.7,0.02,x/100)};
\addlegendentry{\(e=2\%\), \(\bar y\)-branch}

\addplot[black, dashed, line width=0.75pt, domain=6.513:32.850, samples=160]
    {100*barf(0.02,0.7,0.05,x/100)};
\addlegendentry{\(e=5\%\), \(\bar y\)-branch}

\addplot[black, dotted, line width=0.90pt, domain=12.430:49.813, samples=160]
    {100*barf(0.02,0.7,0.10,x/100)};
\addlegendentry{\(e=10\%\), \(\bar y\)-branch}

\end{axis}
\end{tikzpicture}\label{fig:tss-level-psi}
\end{figure}

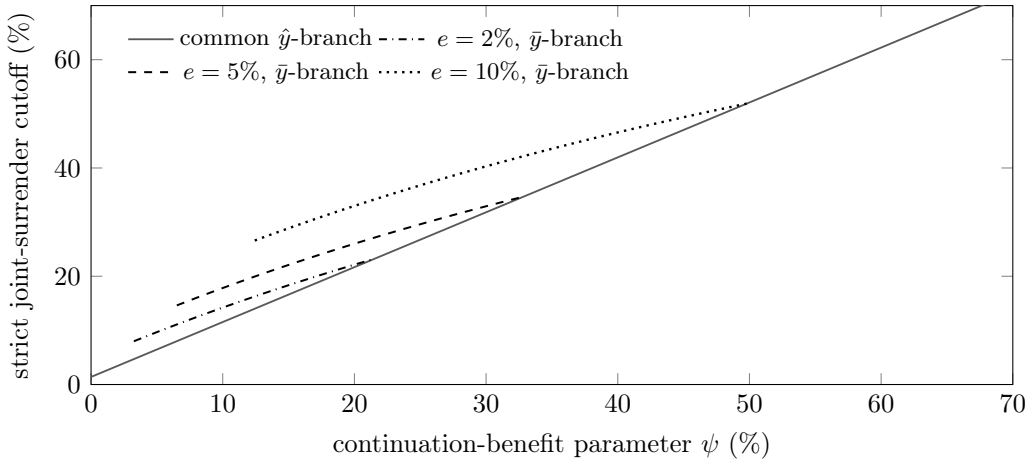
\noindent Figure 2: Strict joint-surrender cutoff over the continuation-benefit range. For each \(e\), the relevant cutoff starts at the corresponding \(\psi'(e,\lambda)\). The common \(\hat y\)-branch is drawn over the full range. The strategic branch binds up to \(\psi^\times(e)\), after which the \(\hat y\)-branch binds.



\begin{figure}[H]
\centering
\begin{tikzpicture}
\begin{axis}[
    thresholdaxis,
    width=0.94\textwidth,
    height=0.45\textwidth,
    grid=none,
    xmin=0.5, xmax=10,
    ymin=0, ymax=30,
    ytick={0,10,20,30},
    xlabel={own-funds ratio \(e=E/(E+2)\) (\%)},
    ylabel={continuation-benefit parameter \(\psi\) (\%)},
    legend columns=2,
    legend style={
        at={(0.50,-0.22)},
        anchor=north,
        font=\footnotesize,
        draw=none,
        fill=none
    },
]

\addplot[gray!70!black, densely dotted, line width=0.80pt, domain=0.5:10, samples=220]
    {100*psip(0.02,0.7,x/100)};
\addlegendentry{lower bound \(\psi'(e)\)}

\addplot[black, solid, line width=0.75pt, domain=0.5:3.839, samples=220]
    {100*psirunf(0.02,0.7,x/100,0.12)};
\addlegendentry{\(10\) pp shock}

\addplot[black, dashed, line width=0.75pt, domain=0.5:6.041, samples=220]
    {100*psirunf(0.02,0.7,x/100,0.17)};
\addlegendentry{\(15\) pp shock}

\addplot[black, dotted, line width=0.90pt, domain=0.5:10, samples=220]
    {100*psirunf(0.02,0.7,x/100,0.27)};
\addlegendentry{\(25\) pp shock}
\end{axis}
\end{tikzpicture}
\end{figure}

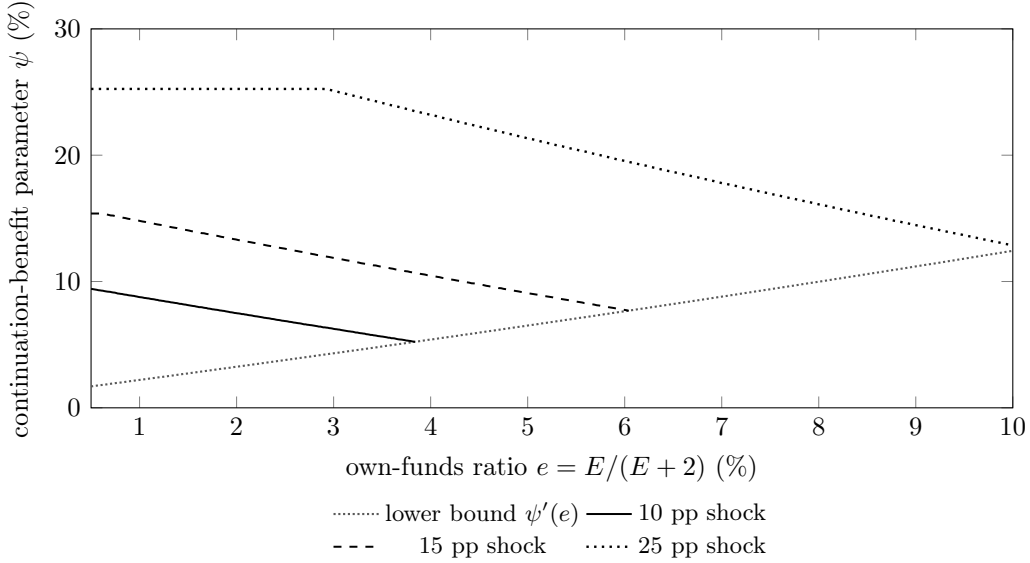
\noindent Figure 3: Inverted strict joint-surrender boundaries for fixed rate shocks. For a given shock, points on or above \(\psi'(e)\) and below the shock boundary satisfy \(y_1>T_{\mathsf{SS}}(e,\psi)\) and are thus in the strict joint-surrender region. A 10-point shock implies strict joint surrender only for thin capitalization and low benefits. Larger shocks expand the run region to higher own-funds and \(\psi\).\medskip


\section{Conclusion}\label{conclusion}

This paper studies mass-surrender risk in life insurance under interest-rate shocks. A higher market rate lowers the interim market value of the fixed-income assets backing outstanding contracts and raises the outside return available to policyholders. The analysis separates the insurer-side underfunding threshold from the equilibrium surrender thresholds faced by policyholders. The former identifies the rate at which the market value of assets no longer covers aggregate contractual surrender claims. Joint surrender, by contrast, is an equilibrium outcome determined by policyholders' payoff comparisons, and hence by asset values, contractual payout rules, outside returns, continuation benefits, and other policyholders' actions.

The model isolates this distinction in a two-player surrender game. In the intermediate underfunding region, one surrender claim can be paid in full but both cannot be paid simultaneously. A surrender by one policyholder can therefore dilute the other policyholder's continuation payoff. Because continuation benefits are proportional to the realized continuation payoff, this dilution also lowers the payoff-equivalent value of remaining in the contract and may thus create strategic complementarity.

The continuation-benefit parameter captures the payoff-dependent value of retained contract features, notably biometric and supplementary insurance cover. The lower bound on this parameter rules out strict surrender incentives under full market-value funding; above it, the surrender-against-continuation threshold lies strictly above the underfunding threshold.

Payoff-dependent continuation benefits have two effects. They raise surrender thresholds by increasing the value of continuation, but because they scale with the realized continuation payoff, they also transmit dilution of the residual asset pool into the value of continuation. When the best-response thresholds cross, this creates a coordination region with both joint continuation and joint surrender equilibria. In the parametrization, such regions arise only at relatively high continuation-benefit levels and sufficiently large effective rate shocks.

The comparative statics show that a larger own-funds buffer weakly raises the strict joint-surrender cutoff, while a larger continuation-benefit parameter raises it strictly. Profit participation has a branch-dependent effect because it moves the return-comparison and dilution branches in opposite directions. The parametrization maps these objects into effective interest-rate and asset-loss thresholds and illustrates that joint surrender depends on the joint configuration of capitalization, market-value losses, credited returns, contractual claims, and perceived continuation benefits.

The framework is structural in its mapping from interest-rate shocks, insurer assets, and contractual payout rules into policyholder payoffs, but restricts strategic interaction to a static two-player game. It abstracts from contract heterogeneity other than size, dynamic surrender waves, retained surrender charges, liquidity management, supervisory intervention, and endogenous asset sales. These restrictions allow the paper to derive closed-form thresholds and isolate the mechanism linking market-value underfunding, payoff-dependent continuation benefits, and strategic surrender incentives. Extensions with heterogeneous policyholder groups, dynamic surrender, or global-game selection could study partial runs, sequential exit, and equilibrium selection.

\appendix

\section{Proofs}

\subsection{Proof of Lemma \ref{lemma1}}\label{sec:proof_lemma1}

\begin{proof}
Fix $i\in\{A,B\}$ and let $j\neq i$. First suppose that $j$ chooses continuation. Under \eqref{eq:BR-def}, surrender is selected if and only if \(V_i^S(\alpha_j)>V_i^C(\alpha_j)\). In the intermediate region \eqref{game.eq.1}, \(x_{i,1}(\textsf{C})=S_{i,1}\) and \(x_{i,2}(\textsf{C})=S_{i,2}\). We thus have \(V_i^S(\textsf{C})>V_i^C(\textsf{C})\) if and only if \(S_{i,1}(1+y_1)>S_{i,2}(1+\psi)\), which is equivalent to \(1+y_1>(1+\psi)(1+\lambda y_0)\), or \(y_1>\hat y(\psi)\).

Now suppose that \(j\) chooses surrender. In the intermediate region,
\[
    x_{i,1}(\textsf{S})
    =
    \frac{S_{i,1}}{S_{A,1}+S_{B,1}}A_1,
\]
and therefore
\[
    V_i^S(\textsf{S})
    =
    \frac{S_{i,0}}{S_{A,0}+S_{B,0}}X_2.
\]
If \(i\) continues while \(j\) surrenders, then
\[
    x_{i,2}(\textsf{S})
    =
    \min\{S_{i,2},X_2-S_{j,1}(1+y_1)\}.
\]
The contractual cap binds if and only if \(y_1\leq\tilde y_i\). In that region,
\[
    V_i^C(\textsf{S})-V_i^S(\textsf{S})
    =
    S_{i,0}(1+\lambda y_0)^2(\psi-\psi')\geq0,
\]
so surrender is not selected. For \(y_1>\tilde y_i\), the continuation payoff is impaired, and
\[
    V_i^S(\textsf{S})>V_i^C(\textsf{S})
    \quad\Longleftrightarrow\quad
    y_1>\bar y_i(\psi).
\]
Finally, \(\bar y_i(\psi)-\tilde y_i\geq0\) under \(\psi\geq\psi'\). Hence the capped region never generates strict surrender, and the impaired-region comparison gives \(\text{\upshape BR}_i(\textsf{\upshape S})=\textsf{\upshape S}\) if and only if \(y_1>\bar y_i(\psi)\).
\end{proof}

\subsection{Proof of Proposition \ref{proposition2}}\label{sec:proof_prop2}

\begin{proof}
By Lemma~\ref{lemma1},
\[
    \text{\upshape BR}_i(\textsf{\upshape C})=\textsf{\upshape S}
    \Longleftrightarrow y_1>\hat y,
    \qquad
    \text{\upshape BR}_i(\textsf{\upshape S})=\textsf{\upshape S}
    \Longleftrightarrow y_1>\bar y_i.
\]
The four equilibrium conditions follow by applying these two best-response equivalences to the definition of Nash equilibrium in \eqref{eq:NE-def}, with continuation selected whenever the relevant strict inequality fails.
\end{proof}

\subsection{Proof of Theorem \ref{theorem1}}\label{sec:proof_theorem1}

\begin{proof}

Fix primitives satisfying the assumptions of Theorem~\ref{theorem1}. Maturity solvency implies
\[
    \psi'
    =
    \frac{X_2}{(S_{A,0}+S_{B,0})(1+\lambda y_0)^2}-1
    \geq0.
\]
Moreover,
\[
    y^*
    =
    \frac{X_2}{(S_{A,0}+S_{B,0})(1+\lambda y_0)}-1,
    \qquad
    y^{\max}
    =
    \frac{X_2}{\max\{S_{A,0},S_{B,0}\}(1+\lambda y_0)}-1,
\]
and \(\hat y(\psi^U)=y^{\max}\) by definition of \(\psi^U\). Since
\[
    \max\{S_{A,0},S_{B,0}\}<S_{A,0}+S_{B,0},
\]
we have \(y^*<y^{\max}\) and hence \(\psi^U>\psi'\).

By Proposition~\ref{proposition2}, both diagonal profiles are Nash equilibria under the continuation tie-break rule at some \(y_1\in(y^*,y^{\max}]\) if and only if there exists an admissible \(\psi\geq\psi'\) such that \(\hat Y\) is nonempty. We prove that this occurs if and only if
\(
    \bar y(\psi^U)<y^{\max}.
\)
Suppose this condition holds.

Since \(\hat y(\psi^U)=y^{\max}\) and \(y^*<y^{\max}\), we have
\[
    \max\{y^*,\bar y_A(\psi^U),\bar y_B(\psi^U)\}
    <
    \hat y(\psi^U).
\]
Choose any
\[
    y_1
    \in
    \left(
    \max\{y^*,\bar y_A(\psi^U),\bar y_B(\psi^U)\},
    \hat y(\psi^U)
    \right).
\]
Then \(y_1\in(y^*,y^{\max}]\), \(y_1<\hat y(\psi^U)\), and \(y_1>\max\{\bar y_A(\psi^U),\bar y_B(\psi^U)\}\).
By Proposition~\ref{proposition2}, \((\textsf{C},\textsf{C})\) and \((\textsf{S},\textsf{S})\) are both Nash equilibria under the continuation tie-break rule at \(y_1\). This proves sufficiency.

For necessity, suppose there exist \(\psi\geq\psi'\) and \(y_1\in(y^*,y^{\max}]\) such that both diagonal profiles are Nash equilibria.
By Proposition~\ref{proposition2}, \(y_1\leq\hat y\), and \(y_1>\bar y\).
Since \(y_1\leq y^{\max}\), this implies, for each \(i\in\{A,B\}\),
\begin{equation}
    \bar y_i<\min\{\hat y, y^{\max}\}.\label{eq:proof-theorem1-necessary-bound}
\end{equation}
Fix \(i\in\{A,B\}\) and let \(j\neq i\). Define
\[
    D_i:=\bar y_i-\hat y.
\]
Using the definitions of \(\bar y_i\) and \(\hat y\),
\[
    \frac{\partial\bar y_i}{\partial\psi}
    =
    \frac{
    X_2S_{i,0}
    }
    {
    S_{j,1}(S_{A,0}+S_{B,0})(1+\psi)^2
    }
    >0,
\]
and
\[
    \frac{\partial^2D_i}{\partial\psi^2}
    =
    -
    \frac{
    2X_2S_{i,0}
    }
    {
    S_{j,1}(S_{A,0}+S_{B,0})(1+\psi)^3
    }
    <0.
\]
Thus \(\bar y_i\) is strictly increasing and \(D_i\) is strictly concave. Moreover, at the lower bound,
\[
    D_i(\psi')
    =
    \bar y_i(\psi')-\hat y(\psi')
    =
    \frac{
    S_{i,0}\left(X_2-S_{A,2}-S_{B,2}\right)
    }
    {
    S_{j,0}(1+\lambda y_0)(S_{A,0}+S_{B,0})
    }
    \geq0,
\]
where the inequality follows from maturity solvency. Finally,
\[
    \lim_{\psi\to\infty}D_i(\psi)=-\infty,
\]
because \(\bar y_i\) is bounded above, whereas \(\hat y(\psi)\to+\infty\). Since \(D_i\) is concave, the set
\[
\{\psi \ge \psi' : D_i\ge 0\}
\]
is an interval, and its complement is an upper interval.

There are two cases. If \(\psi\geq\psi^U\), monotonicity of \(\bar y_i\) and \eqref{eq:proof-theorem1-necessary-bound} give
\[
    \bar y_i(\psi^U)\leq \bar y_i(\psi)<y^{\max}.
\]
If \(\psi<\psi^U\), then \(\hat y(\psi)<y^{\max}\), and \eqref{eq:proof-theorem1-necessary-bound} gives \(D_i(\psi)<0\). Since \(D_i\) is concave, \(D_i(\psi')\geq0\), and \(D_i(\psi)\to-\infty\), negativity persists to the right of the crossing; hence \(D_i(\psi^U)<0\). Using \(\hat y(\psi^U)=y^{\max}\), this gives \(\bar y_i(\psi^U)<y^{\max}\). The argument applies to both \(i=A,B\). Hence
\[
    \max\{\bar y_A(\psi^U),\bar y_B(\psi^U)\}<y^{\max}.
\]
\end{proof}

\subsection{Proof of Theorem \ref{theorem2}}\label{sec:proof_theorem2}

\begin{proof}
By Proposition~\ref{proposition2}, \((\textsf{S},\textsf{S})\) is the unique Nash equilibrium at some \(y_1\in(y^*,y^{\max}]\) if and only if, for some \(\psi\geq\psi'\),
\[
    \max\{\hat y(\psi),\bar y_A(\psi),\bar y_B(\psi)\}<y^{\max}.
\]
For \(\psi\geq\psi'\), all three thresholds are increasing in \(\psi\). Hence the left-hand side is minimized at \(\psi=\psi'\). At this lower bound, \(\hat y(\psi')=y^*\) and \(\bar y_i(\psi')=\tilde y_i\) for \(i=A,B\). Since \(y^*<y^{\max}\), the preceding inequality holds for some admissible \(\psi\) if and only if
\(
    \tilde y<y^{\max}.
\)
\end{proof}

\subsection{Proof of Proposition \ref{prop:comp_stat_surrender_cutoff}}\label{sec:proof_prop3}

\begin{proof}
Let
\[
        s:=S_{A,0}+S_{B,0},\qquad
        r:=1+y_0,\qquad
        q:=1+\lambda y_0,\qquad
        z:=1+\psi .
\]
Under the maintained assumptions, \(r>0\), \(q>0\), and \(z>0\). We also use
\[
        X_2=(E+s)r^2 .
\]

Since maturity solvency gives \(X_2\geq s q^2\), the lower bound satisfies
\[
    \psi'=\frac{X_2}{s q^2}-1\geq0.
\]
Hence, under \(\psi\geq \psi'\), \(z=1+\psi\geq1\).

For \(i\in\{A,B\}\) and \(j\neq i\), write \(S_i:=S_{i,0}\) and \(S_j:=S_{j,0}\). The best-response threshold for surrender against continuation is
\[
        \hat y=zq-1,
\]
whereas the best-response threshold for surrender against surrender for policyholder \(i\) is
\[
        \bar y_i
        =
        \frac{X_2}{S_j q}
        \left(1-\frac{S_i}{s z}\right)-1
        =
        \frac{(E+s)r^2}{S_j q}
        \left(1-\frac{S_i}{s z}\right)-1 .
\]

Within the intermediate region, the joint-surrender equilibrium is unique whenever
\[
        y_1>T_{\textsf{SS}}.
\]

Since \(q=(1-\lambda)+\lambda r\), \(\lambda\in(0,1)\), and \(r>0\), we have \(q>0\). Also \(z=1+\psi\geq1\). For \(i\in\{A,B\}\) and \(j\neq i\),
\[
        1-\frac{S_i}{sz}
        =
        \frac{s z-S_i}{s z}
        =
        \frac{S_j+\psi s}{s z}
        >0.
\]
The derivatives of \(\hat y=zq-1\) are
\[
        \frac{\partial \hat y}{\partial \psi}=q,
        \qquad
        \frac{\partial \hat y}{\partial y_0}=z\lambda,
        \qquad
        \frac{\partial \hat y}{\partial \lambda}=z y_0.
\]
This proves the stated monotonicity of \(\hat y\).

For the surrender-against-surrender branch,
\[
        \bar y_i+1
        =
        \frac{(E+s)r^2}{S_jq}
        \left(1-\frac{S_i}{s z}\right).
\]
Differentiating with respect to \(E\), \(\psi\), \(y_0\), and \(\lambda\), while holding the other primitives fixed, gives
\[
        \frac{\partial \bar y_i}{\partial E}
        =
        \frac{r^2}{S_jq}
        \left(1-\frac{S_i}{s z}\right)>0,
\]
\[
        \frac{\partial \bar y_i}{\partial \psi}
        =
        \frac{(E+s)r^2S_i}{S_jqs z^2}>0,
\]
\[
        \frac{\partial \bar y_i}{\partial y_0}
        =
        \frac{(E+s)r(2-\lambda+\lambda y_0)}{S_jq^2}
        \left(1-\frac{S_i}{s z}\right)>0,
\]
where \(2-\lambda+\lambda y_0=q+1-\lambda>0\), and
\[
        \frac{\partial \bar y_i}{\partial \lambda}
        =
        -
        \frac{(E+s)r^2y_0}{S_jq^2}
        \left(1-\frac{S_i}{s z}\right).
\]
Thus each \(\bar y_i\) is strictly increasing in \(E\), \(\psi\), and \(y_0\). Each \(\bar y_i\) is strictly decreasing in \(\lambda\) if \(y_0>0\), strictly increasing in \(\lambda\) if \(y_0<0\), and independent of \(\lambda\) if \(y_0=0\). Since \(\bar y=\max\{\bar y_A,\bar y_B\}\), the same monotonicity statements hold for \(\bar y\).

The cutoff \(T_{\textsf{SS}}=\max\{\hat y,\bar y_A,\bar y_B\}\) is a finite maximum. A finite maximum of weakly increasing functions is weakly increasing, and a finite maximum of strictly increasing functions is strictly increasing. Because \(\hat y\) is independent of \(E\) and \(\bar y_A,\bar y_B\) are strictly increasing in \(E\), \(T_{\textsf{SS}}\) is weakly increasing in \(E\). Because all three component thresholds are strictly increasing in \(\psi\) and \(y_0\), \(T_{\textsf{SS}}\) is strictly increasing in \(\psi\) and \(y_0\).

The \(\lambda\)-comparative static of \(T_{\mathsf{SS}}\) is branch-dependent. When \(\hat y\) is the unique active branch, the local derivative has the sign of \(y_0\). When a strategic branch \(\bar y_i\) is uniquely active, the local derivative has the opposite sign. At branch crossings, \(T_{\mathsf{SS}}\) may be kinked. Therefore, for \(y_0\neq0\), \(T_{\mathsf{SS}}\) has no unconditional monotone comparative static in \(\lambda\).

It remains to prove the policy-size statements. Recall \(L:=\max\{S_{A,0},S_{B,0}\}\), \(\theta:=L/s\) and let 
\[
     M:=\min\{S_{A,0},S_{B,0}\}.
\]
Since the common factor \((E+s)r^2/q\) is positive,
\[
\begin{aligned}
        \bar y_A-\bar y_B
        &=
        \frac{(E+s)r^2}{q}
        \left[
        \frac{1}{S_{B,0}}
        \left(1-\frac{S_{A,0}}{s z}\right)
        -
        \frac{1}{S_{A,0}}
        \left(1-\frac{S_{B,0}}{s z}\right)
        \right] \\
        &=
        \frac{\psi(S_{A,0}-S_{B,0})(E+s)r^2}
             {S_{A,0}S_{B,0}zq} .
\end{aligned}
\]
Hence the larger initial policy determines \(\bar y\) when \(\psi>0\), and the two branches coincide when \(\psi=0\). Therefore
\[
        \bar y+1
        =
        \frac{(E+s)r^2}{Mq}
        \left(1-\frac{L}{s z}\right).
\]
With \(L=\theta s\) and \(M=(1-\theta)s\), this becomes
\[
        \bar y+1
        =
        \frac{(E+s)r^2}{s q}
        \frac{1-\theta/z}{1-\theta} .
\]
Holding \(s\) fixed,
\[
        \frac{\partial \bar y}{\partial \theta}
        =
        \frac{(E+s)r^2}{s q}
        \frac{\psi}{z(1-\theta)^2}
        \geq0,
\]
with strict inequality if \(\psi>0\). At \(\theta=1/2\), this is the right derivative on the admissible domain \(\theta\in[1/2,1)\). Holding \(\theta\) and \(E\) fixed,
\[
        \frac{\partial \bar y}{\partial s}
        =
        -
        \frac{Er^2}{q s^2}
        \frac{1-\theta/z}{1-\theta}
        \leq0,
\]
with strict inequality if \(E>0\).
\end{proof}

\printbibliography

\end{document}